\documentclass[reqno]{llncs}
\usepackage{framed}
\usepackage{url}
\usepackage[ruled,vlined,linesnumbered]{algorithm2e}
\usepackage{comment}
\usepackage{tikz}
\usepackage[utf8]{inputenc}
\usepackage{array}
\usepackage{color, colortbl}
\usepackage[T1]{fontenc}
\usepackage{todonotes}
\usepackage{float}

\usepackage{soul}
\usepackage{amssymb}
\usepackage{multirow}
\usepackage{pdflscape}
\usepackage[export]{adjustbox}
\usepackage{amsmath}
\usepackage{multirow}
\usepackage{threeparttable}
\usetikzlibrary{arrows,shapes}
\newcommand{\gb}{\textsc{Graph Burning}}

\newcommand{\GC}{\textbf{\scriptsize{GET-CENTER($G,b$) }}}
\newcommand{\GCT}{\textbf{\scriptsize{GET-CENTERS-FOR-MULTI-ROOTED-TREE($T,b$)}}}
\newcommand{\GCST}{\textbf{\scriptsize{GET-CENTERS-FOR-SINGLE-ROOTED-TREE($T,b$)}}}
\newcommand{\BGC}{\textbf{\scriptsize{BURN-GUESS-CACTUS($G,b$) }}}
\newcommand{\AR}{\textbf{\scriptsize{ARTICULATION-POINT($G,f,r$) }}}
\newcommand{\RA}{\textbf{\scriptsize{RANDOM-ARTICULATION($G$) }}}

\newcommand{\GR}{\textbf{\scriptsize{$G(V,E)$ }}}
\newcommand{\TG}{\textbf{\scriptsize{$T(V,E)$ }}}

\newcommand{\OD}{\textbf{\scriptsize{$D^+$ }}}
\newcommand{\ID}{\textbf{\scriptsize{$D^-$ }}}
\newcommand{\BC}{\textbf{\scriptsize{$b$-CUTTING}}}

\newcommand{\BADG}{\textbf{\scriptsize{BAD-GUESS}}}

\newcommand{\merge}{\textbf{\scriptsize{MERGE-AND-BURN}}}
\tikzset{
	process1/.style={rectangle, minimum width=8cm, minimum height=1cm, text width=10cm,text centered, draw=black},
	process2/.style={rectangle, minimum width=4cm, minimum height=1cm, text width=5cm,text centered, draw=black},
	decision/.style={diamond, text centered, draw=black, aspect=2, inner xsep=-2mm},
	stop/.style={rectangle, rounded corners, minimum width=2cm, minimum height=1cm,text centered, draw=black},
	decision1/.style={diamond, text centered, draw=black, aspect=2, inner xsep=-2mm},
	stop/.style={rectangle, rounded corners, minimum width=2cm, minimum height=1cm,text centered, draw=black},
	arr/.style={thick,-stealth}
}

\newcounter{cases}
\newcounter{subcases}[cases]

\DeclareMathOperator*{\argmax}{argmax}

\urldef{\mails}\path|19mcpc06@uohyd.ac.in,askcs@uohyd.ac.in,sdbcs@uohyd.ac.in|

\SetKwComment{Comment}{$\triangleright$\ }{}

\title{Approximation Algorithms for the Graph Burning on Cactus and  Directed Trees}

\author{Rahul Kumar Gautam \and Anjeneya Swami Kare \and S. Durga Bhavani}
\institute{School of Computer and Information Sciences, \\University of Hyderabad,\\ Hyderabad, India\\ \mails}

\begin{document}
\tikzstyle{vertex}=[circle,fill=black!25,minimum size=12pt,inner sep=0pt]
\tikzstyle{selected vertex} = [vertex, fill=red!60]
\tikzstyle{edge} = [draw,thick,-]
\tikzstyle{weight} = [font=\small]
\tikzstyle{selected edge} = [draw,line width=3pt,-,red!50]
\tikzstyle{ignored edge} = [draw,line width=3pt,-,black!20]
\maketitle

\begin{abstract}
	Given a graph $G=(V, E)$, the problem of \gb{} is to find a sequence of nodes from $V$, called a burning sequence, to burn the whole graph. This is a discrete-step process, and at each step, an unburned vertex is selected as an agent to spread fire to its neighbors by marking it as a burnt node.  A burnt node spreads the fire to its neighbors at the next consecutive step.    The goal is to find the burning sequence of minimum length. The \gb{} problem is NP-Hard for general graphs and even for binary trees. A few approximation results are known, including a $ 3$-approximation algorithm for general graphs and a $ 2$-approximation algorithm for trees.
	The \gb{} on directed graphs is more challenging than on undirected graphs. In this paper,  we propose  1) A $2.75$-approximation algorithm for a cactus graph (undirected),  2) A $3$-approximation algorithm for multi-rooted directed trees (polytree) and 3) A $1.905$-approximation algorithm for single-rooted directed tree (arborescence). We implement all the three approximation algorithms and the results are shown for randomly generated cactus graphs and directed trees.
\end{abstract}


\section{Introduction}\label{sec:intro}
In a public information campaign, choosing people who can spread authentic information within a village network in the least possible time is important. This problem is modeled as \gb{} problem,  which is defined as  a discrete step sequential process. One node is selected at each time-step $t$ as a source of fire which burns all of its neighbors at the next time step $t+1$. Another node is picked as the source of fire at time $t+1$, and the process of burning continues till all the nodes of the graph are burned.  The burned nodes can burn their neighbors at the next step. The information can be passed from informed nodes in the previous time step. The sequence of nodes selected as fire sources is a burning sequence. An informed node remains in an informed or burned state throughout the process. The goal is to compute the burning sequence of minimum length. The length of the optimal burning sequence is called the  \textit{burning number} of the graph, denoted by $b(G)$. Note that $b(G)$ is a positive integer. The \gb{} problem is similar to the $k$-centre problem~\cite{Garcia2018kcenter}, in which all
the sources of fire are selected in the first step itself. The applications of \gb{} problem can be seen in political campaigns~\cite{2020heuristics} and satellite communication~\cite{Simon2019heuristic}.


Let $G(V, E)$ be a finite, connected, unweighted, and undirected graph having a set of vertices $V$  and edges $E$. Let $S=\{v_0,v_1,v_2,\cdots v_{b-1}\}$ be a burning sequence. At time $t_0$, the vertex $v_0$ is burnt. At time $t_1$, the vertex $v_1$ is burnt, and the neighbors of vertex $v_0$ are also burnt. Similarly, at time $t_i$, the vertex $v_i$ is burnt, and all the vertices which caught fire at the time $t_{i-1}$ will burn their unburned neighbors. The process continues until all the nodes of the graph are burnt. The burning process for the directed graphs is very similar, except that a node burns its out-neighbors.

In this paper, we propose the following approximation algorithms:
\begin{enumerate}
	\item A $2.75$-approximation algorithm for the cactus graph.
	\item A $3$-approximation algorithm for the multi-rooted directed trees. This leads to a $2$-approximation algorithm for single-rooted directed trees.
	\item A $1.905$-approximation algorithm for the single-rooted directed trees.
\end{enumerate}



\section{Related Work}

Bonato et al.~\cite{Bonato2014Contagion} proposed the \gb{} problem. They studied characteristics and proposed bounds for the \gb{} problem. Bonato et al.~\cite{Bonato2016burn} and Bessy et al.~\cite{Bessy2017burning} proved that the \gb{} problem is NP-Complete for general undirected graphs, binary trees, spider graphs, and path forests. Bessy et al.~\cite{Bessy2018bounds} proved bounds for general graphs and trees.
For a connected graph $G$ of order $n$, Bonato et al. ~\cite{Bonato2016burn}  gave upper bound $b(G) \le 2\sqrt{n}-1$ and conjectured an upper bound of $\lceil \sqrt{n} \rceil$. Bessy et al.~\cite{Bessy2018bounds} improved the bounds for connected graphs and trees.  For connected graph $G$ with $n$ vertices,   for  $0<\epsilon <1$, $b(G)\le \sqrt{\frac{32}{19}. \frac{n}{1-\epsilon}} +  \sqrt{\frac{27}{19\epsilon}}$ and $b(G) \le \sqrt{\frac{12n}{7}}+3$. For a tree $T$ of order $n$, if $n_2$ is the number of vertices of degree $2$  and $n_{\ge 3}$ is the number of vertices with the degree at least three, then $b(T) \le \lceil \sqrt{(n+n_2)+ \frac{1}{4}+ \frac{1}{2}} \rceil$ and $b(T) \le \lceil \sqrt{n} \rceil + n_{\ge 3}$. Land and Lu~\cite{Land2016upper} improved the upper for connected graphs to $\frac{\sqrt{6}}{2}\sqrt{n}$.

The \gb{} problem is fixed-parametrized tractable (FPT) for parameters: distance to cluster and neighborhood diversity~\cite{Kare2019parameterized}. Kare et al.~\cite{Kare2019parameterized} proved that the problem is polynomial-time solvable for   cographs and split graphs. Later, Yasuaki et al.~\cite{2021parameterized} proved that the \gb{} problem is $W[2]$-Complete for  the parameter $b(G)$ and is FPT for the combined parameter of clique-width and the maximum diameter among the connected components.
The \gb{} problem was studied for different classes of graphs such as spider graphs~\cite{Bonato2019spider}, generalized Peterson graphs~\cite{Kai2017peterson}, dense graphs~\cite{Kamali2020dense} and theta graphs~\cite{Lui2019theta}. Further, several heuristics~\cite{Simon2019heuristic,Farokh2020NewHF,2021heuristics} have been proposed based on centrality measures and the graph structure for the \gb{} problem.
A generalization of \gb{} problem called $k$-burning problem is studied in~\cite{2021kburning}.

Bonato et al.~\cite{Bonato2019approx} proposed a $3$-approximation algorithm for general graphs, a $2$-approximation algorithm for trees, and a $1.5$-approximation algorithm for disjoint paths. Bonato et al.~\cite{Bonato2019approx} also showed that the \gb{} problem could be solved in polynomial time on path forests, where the number of paths is constant. Further, they proposed approximation schemes for the path forests where the number of paths is not constant. Recently, Diaz et al.~\cite{2021Approximation} proposed an approximation algorithm with approximation factor $3- 2/b(G)$ for general graphs.

Janssen et al. \cite{2020directedGraph} studied the \gb{} problem for directed graphs. They studied the problem for directed acyclic graphs and directed trees.  Janssen et al. \cite{2020directedGraph} showed that \gb{} problem is NP-Complete for directed trees and directed acyclic graphs (DAG). They proved that the \gb{} problem  is FPT for trees with  burning number  as the parameter and W[2]-Hard for DAG with  burning number  as the parameter.



The rest of the paper is organized as follows:  In section \ref{sec:cac},  we discuss the approximation algorithm for the cactus graph. In section \ref{sec:tree}, we present  $3$-approximation algorithm for polytrees (directed multi-rooted tree). Section \ref{sec:dtree} discusses the $1.905$-approximation algorithm for  directed trees (arborescence).

For $b \in Z^+$, the sources of fire have the corresponding burning ranges $[b-1,b-2,b-3, \cdots 0]$. Throughout the paper, for any decimal value $k$, by burning range $bk$, we mean, $\lceil b*k \rceil$.

\section{Approximation algorithm for cactus}
\label{sec:cac}
Bonato et al.~\cite{Bonato2019approx} proposed a $3$-approximation algorithm for the general undirected graphs and a $2$-approximation algorithm for trees.  A cactus graph is an undirected graph in which any pair of simple cycles share at most one vertex.  A cactus graph is structurally close to a tree. 

The $3$-approximation and $(3- 2/b(G))$-approximation algorithms for general graphs apply to the cactus. There is no known less than $3$-approximation algorithm for graphs that contain cycles. This paper proposes a $ 2.75$-approximation algorithm for the cactus graph.
The algorithm takes a cactus graph $G$ and a positive integer $b$ as input and returns a burning sequence of length $2.75b$ or returns \BADG{} $(-1)$ to indicate that the graph can not be burned in at most $b$ rounds.

If the cactus graph has no articulation point, then the graph is a single cycle, and hence the problem is polynomial time-solvable~\cite{Bonato2014Contagion}.  In the rest of the paper, we assume that the cactus graph has at least one articulation point.
The algorithm is as follows: For burning number $2.75b$, we have the sequence $BR$, which contains $b$ number of burning ranges as shown below.  $$BR = [\underbrace{1.75b, 1.75b+1, 1.75b+2, \cdots, 2b-1}_{BR_1},\underbrace{2b, 2b+1, \cdots 2.75b-1}_{BR_2}]$$
For the first part $|BR_1| = 0.25b$ and contains burning radii of range at least  $1.75b$.
And for the second part $|BR_2| = 0.75b$ and contains  burning radii of range at least $2b-2$.

\fbox{
\begin{minipage}{32em}
	\textbf{Input} : The graph G=(V,E) and $b \in Z^+$. \\
	\textbf{Output} : Either the algorithm returns a burning sequence or \BADG{} (-1).
\end{minipage}

}
\begin{enumerate}
\item Initially, all the vertices of graph $G$ are unmarked.
\item Initialize, variables $b_1 = |BR_1| = \lceil 0.25b \rceil$, $b_2 = |BR_2| =  \lceil 0.75b \rceil$ and the corresponding  burning sequences $BS_1 = BS_2 = [\;]$ and a random articulation point $r$.
\item If all vertices of the graph are marked,
\begin{itemize}
	\item Return $(BS_1,BS_2)$.
\end{itemize}
\item Select the farthest unmarked node $f$ from $r$.
\begin{enumerate}
	\item \textbf{If} there exists a path $P = \{ v_1 = f, v_2, \cdots, v_k\}$ from the node $f$ towards the root $r$ and $v_k$ is an articulation point such that $0.25b \le d(f, v_k) \le 1.75b$ and $b_1 \ge 1$, then
	\begin{itemize}
		\item Add source of the fire $v_k$ to $BS_1$.
		\item Mark all vertices  $ u$ in $\{ v \in V(G): d(v_k,v)\le 1.75b \}$
		\item Decrease $b_1$ by $1$.
	\end{itemize}
	\item \textbf{Else-if} $b_2 \ge 1$,
	\begin{itemize}
		\item  Add source of the fire $f$ to $BS_2$.
		\item Mark all vertices  $ u$ in  $\{ v \in V(G): d(f,v)\le 2b-2 \}$.
		\item Decrease $b_2$ by $1$.
	\end{itemize}
	\item \textbf{Else}
	\begin{itemize}
		\item  return \BADG{} (-1).
	\end{itemize}
\end{enumerate}
\item Repeat steps 3-4 until all the vertices of graph $G$ get marked and return the burning sequence or return -1.
\end{enumerate}
Initially, all the vertices of graph $G$ are unmarked. We take an arbitrary articulation point $r \in V$ as a root.
Let the unmarked farthest node from $r$ be $f$, and  $\text{N}_x[f]$ be the union of all \textbf{unmarked} vertices around $f$ up to burning range $x \in Z^+$ in the burning process. Obviously for any $u,v \in V(G)$, if $u \notin \text{N}_{2b-2}[v]$, to burn both the vertices ($u$ and $v$) we need at least $2$ rounds if $b(G)\le b$. The question is: can we mark all $w$ in $\text{N}_{2b-2}[v]$ in less than the $2b-2$ burning range(radius)? Yes, it is possible due to the presence of articulation points in the cactus graph. The importance of articulation points in the graph is shown in \figurename~\ref{fig:1.75b}. If we mark the vertices up to distance $10$ around $f$, then the vertices $5$ and $8$ are not covered. On the other hand, if we mark the vertices around the vertex $v_k$ up to distance $7$, all the vertices are covered, means that $\text{N}_{10}[f] \subseteq \text{N}_{7}[v_k]$. With this intuition, we state the following lemma.

\begin{lemma}
\label{lemma1}
$f$ is the farthest vertex from the root $r$. If there exists a path $P = \{ v_1 = f, v_2, \cdots, v_k\}$ from the node $f$ towards the root $r$ such that $v_k$ is an articulation point and $0.25b \le d(f, v_k) \le 1.75b$, then $\text{N}_{2b-2}[f] \subseteq \text{N}_{1.75b}[v_k]$.
\end{lemma}

\begin{proof}
If $0.25b \le d(f, v_k) \le 1.75b$,  $\text{N}_{1.75b}[f] \subseteq \text{N}_{1.75b}[v_k]$ since $f$ is the farthest unmarked vertex and  $f \in \text{N}_{1.75b}[v_k]$.  
Remaining $ u \in (\text{N}_{2b-2}[f] \backslash \text{N}_{1.75b}[f])$ are within $0.25b$ distance from $v_k$. Therefore, $\text{N}_{2b-2}[f] \subseteq \text{N}_{1.75b}[v_k]$.

\end{proof}

\begin{figure}
\centering
\def\svgwidth{\columnwidth}
\scalebox{0.10}{\includegraphics{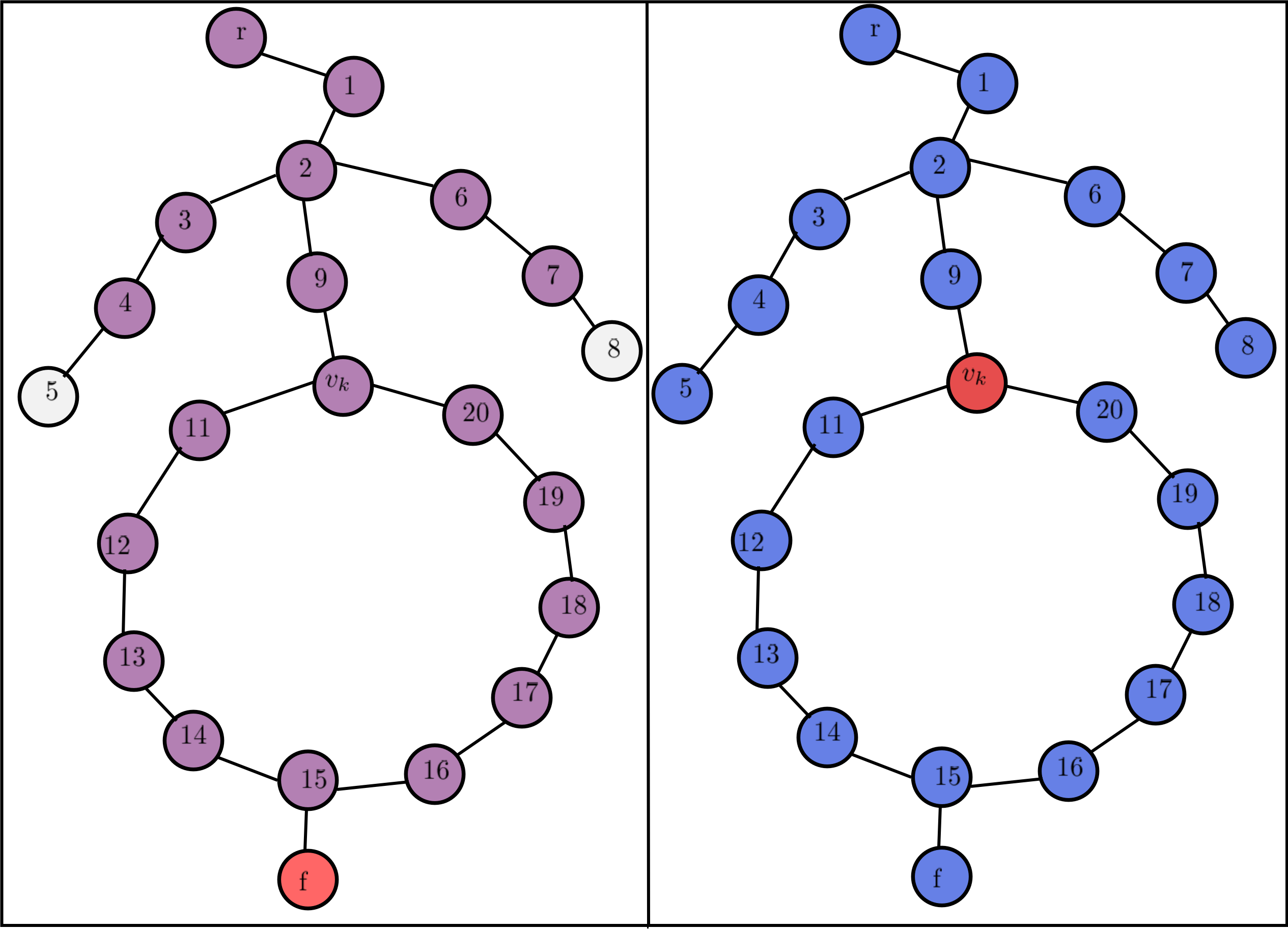}}
\caption{$f$ is a farthest unmarked node from the root $r$ and $v_k$ is an articulation point in the path joining $f$ and $r$. A small radius centered at $v_k$ can cover all unmarked vertices as covered by a larger radius centered at $f$. Clearly, $\text{N}_{10}[f] \subseteq \text{N}_{7}[v_k]$.}
\label{fig:1.75b}
\end{figure}

\begin{lemma}
\label{lemma2}
Let $f$ be the farthest unmarked node from the root $r$. If there is no articulation point $v_k$ towards the root such that $0.25b \le d(f, v_k) \le 1.75b$, there exists a cycle of length at least $3b$.
\end{lemma}

\begin{proof}
If there is no articulation point $v_k$ such as $0.25b \le d(f, v_k) \le 1.75b$, then two parallel paths each of length  at least $1.5b$ exist, and both the paths meet at a point towards the root $r$ from $f$. Therefore there exists a cycle of length at least $3b$. { }
\end{proof}

At each iteration, we find $f$, the farthest unmarked node from the root $r$. If  there exists an articulation point $v_k$ satisfying constraints of Lemma~\ref{lemma2}, then we add $v_k$ to the burning sequence $BS_1$ and mark nodes around $v_k$ up to distance $1.75b$; otherwise, we add $f$ to the burning sequence $BS_2$  and mark nodes around $f$ up to distance $2b-2$.

The process is depicted in  Algorithm~\ref{algo:cactus}. The \BGC procedure takes a cactus graph $G$ and a positive integer $b$ as input and returns either a burning sequence or \BADG{} to indicate that the graph can not be burned in $b$ rounds.

\begin{lemma}
\label{lemma3}
If there are $k$ cycles, each of length at least $3b$, then there are $k$ arcs (paths) each of length $b$ such that the distance between any pair of arcs is at least $b$. Here, we say, two arcs $A_1$ and $A_2$ are at a distance at least  $b$ if for every pair of vertices $v_1 \in A_1$ and $v_2 \in A_2$, $d(v_1, v_2) \ge b$.
\end{lemma}

\begin{proof}
We prove the lemma by induction. The statement is trivially valid for $k = 1$ and $k=2$. For $k=2$, we can form arcs around opposite ends of the cycles.

Let us assume that the statement is valid for the number of cycles $k-1$; we now prove that the statement is true for $k$ cycles too. From any articulation point, consider the farthest node $f$, which is part of a cycle $C$. Consider $b$-length arc centered at $f$. Any vertex which is at distance less than $b$  from this arc cannot be shared by another cycle as that will contradict the fact that $f$ is the farthest node. Therefore any vertex of other cycles will be at least $b$ distance from the arc $C$. Using this argument and induction hypothesis, we get $k$ arcs at a pairwise distance at least $b$.

\end{proof}

\begin{lemma}
\label{lemma4}
If there are $0.75b+1$ cycles of length at least $3b$ then $b(G) > b$
\end{lemma}
\begin{proof}
By Lemma~\ref{lemma3}, there are $0.75b+1$ arcs each of length $b$ such that the arcs are at a pairwise distance of at least $b$. To completely burn an arc with a single round, we need a burn radius of at least $0.5b$.

For the sake of contradiction, let us assume that the graph can be burned in $b$ rounds. There are $0.5b$ number of burn radii of range at least $0.5b$. Therefore, at most $0.5b$ arcs are completely burned with these radii. The remaining arcs are $0.25b+1$, each requiring at least $2$ rounds to burn the arc completely. So the total number of rounds $0.5b + 2(0.25b+1) > b$, which is a contradiction. Therefore $b(G) > b$.

\end{proof}

\begin{lemma}[From~\cite{Bonato2019approx}]
\label{lemma:r2rm1}
For a positive integer $b$, if there are $b$ vertices at a pairwise distance of at least $2b-1$, then $b(G) \ge b$.		
\end{lemma}

From Lemma~\ref{lemma:r2rm1}, we have the following Corollary. The proof of the Corollary is very similar to the proof of Lemma~\ref{lemma:r2rm1} of Bonato et al. ~\cite{Bonato2019approx}.

\begin{corollary}
\label{cor:r2rm1}
For a positive integer $b$, if there are $b+1$ vertices at a pairwise distance of at least $2b-1$ in $G$, then $b(G) \ge b+1$.
\end{corollary}

\begin{proof}
Let $S = \{ x_1, x_2, \cdots, x_{b+1} \}$ be the vertices at a pairwise distance of at least $2b-1$. If there is a burning
sequence of length $b$, each node in the burning sequence spreads fire to the nodes at a distance of at most $b-1$. For each
$x_i \in S$, consider a circle of radius $b-1$.  No two circles intersect, as we know the distance between the centers is at least
$2b-1$.  To burn a center vertex,  we should either put fire at the center vertex or at a vertex in its $b-1$ circle.
Any burned vertex in one circle can not spread fire to another circle's center because the distance between the vertex and the other center
is greater than $b-1$. So we should include at least one vertex from each circle in the burning sequence. Therefore, we can not burn the graph
in $b$ rounds and hence $b(G) \ge b+1$.

\end{proof}

\begin{lemma}
\label{lemma6}
If algorithm \BGC returns \BADG{}, then $b(G) \ge b+1$.
\end{lemma}

\begin{proof}
The algorithm returns \BADG{}, we have two cases:

\begin{description}
\item[Case-1:]
The number of times, the marking radius $2b-2$ is used, is greater than $0.75b$. There are at least $0.75b+1$ cycles of length at least $3b$, by Lemma~\ref{lemma4}, $b(G) \ge b+1$. In  Algorithm~\ref{algo:cactus}, the case occurs when $b_2\le 0$.
\item[Case-2:]
The sum of the number of times marking radius $2b-2$ and $1.75b$  are used is greater than $b$. There are $b+1$ vertices at a pairwise distance of at least $2b-1$  and hence by Corollary~\ref{cor:r2rm1}, $b(G) \ge b+1$. In  Algorithm~\ref{algo:cactus}, the case occurs when $ b_1 + b_2\le 0$.
\end{description}

\end{proof}

In the algorithm, to get a burning sequence, the burning ranges $2b-2$ and $1.75b$ are used $0.75b$ and $0.25b$ times, respectively. To get these many radii, the maximum   length of the burning sequence needed is $2.75b$. With this and Lemma~\ref{lemma6}, we have the following Theorem.
\begin{theorem}
A polynomial time $2.75$-approximation algorithm exists for \gb{} on cactus graphs.
\end{theorem}

\begin{algorithm}[!htbt]
\label{algo:cactus}
\DontPrintSemicolon
\SetKwInOut{Input}{Input}\SetKwInOut{Output}{Output}
\Input   {$G=(V,E) $ and a positive integer $b$.}
\Output  {Returns burning sequence or \BADG{} (-1).}
\SetKwFunction{FMain}{BURN-GUESS}
\SetKwProg{Fn}{Function}{}
\Fn{\FMain{$G$,$b$}} {
\Begin{
$BS_1 \gets [ \; ]$\;
$BS_2 \gets [ \; ]$\;
$b_1 \gets \lceil 0.25b\rceil $\;
$b_2 \gets \lceil 0.75b\rceil $\;
$r \gets \RA$\;
$M \gets \phi$\; \Comment{Set $M$ contains all the marked vertices}
\While{$|M| \ne |V(G)|$}{
	$f \gets \argmax_x\{d(r,x)\; where \; x\notin M\} $\;
	\Comment{ A vertex $f$ is the farthest unmarked vertex from $r$}
	$v_k \gets \AR$\;
	\Comment{If there is an
		articulation point $v_k$ towards the root $r$ such that $0.25b \le d(f, v_k)\le 1.75b$, $\AR$ returns a vertex $v_k$; otherwise $\AR$ returns $-1$.}
	\If{$v_k \ne -1$ and $b_1 \ge 1$}{
		$append(BS_1,v_k)$\;
		$M\gets M \cup \{ v \in V(G): d(v_k,v)\le 1.75b \}$\;
		$b_1 \gets b_1 - 1 $\;
	}
	\ElseIf{$b_2 \ge 1$}{
		$append(BS_2,f)$\;
		$M\gets M \cup \{ v \in V(G): d(f,v)\le 2b-2 \}$\;
		$b_2 \gets b_2 -1 $\;
	}
	\Else{
		\KwRet{$-1$}
	}
	
}
\KwRet{$(BS_1,BS_2)$}
\Comment{concatenation of sequence $BS_1$ with $BS_2$.}			
}

}
\caption{Approximation algorithm for cactus graph \GR.}
\end{algorithm}
\section{$3$-approximation algorithm for directed trees}
\label{sec:tree}

No known approximation algorithms exist for graph burning on directed graphs or directed trees. The $2$-approximation algorithm of undirected trees will not scale to directed trees as shown in \figurename~\ref{fig:2approxdtree}. This section presents a $3$-approximation algorithm for graph burning on directed trees (polytree). The algorithm leads to a $ 2$-approximation algorithm for graph burning on the single-rooted tree (arborescence).

\begin{figure}
\centering
\def\svgwidth{\columnwidth}
\scalebox{0.15}{\includegraphics{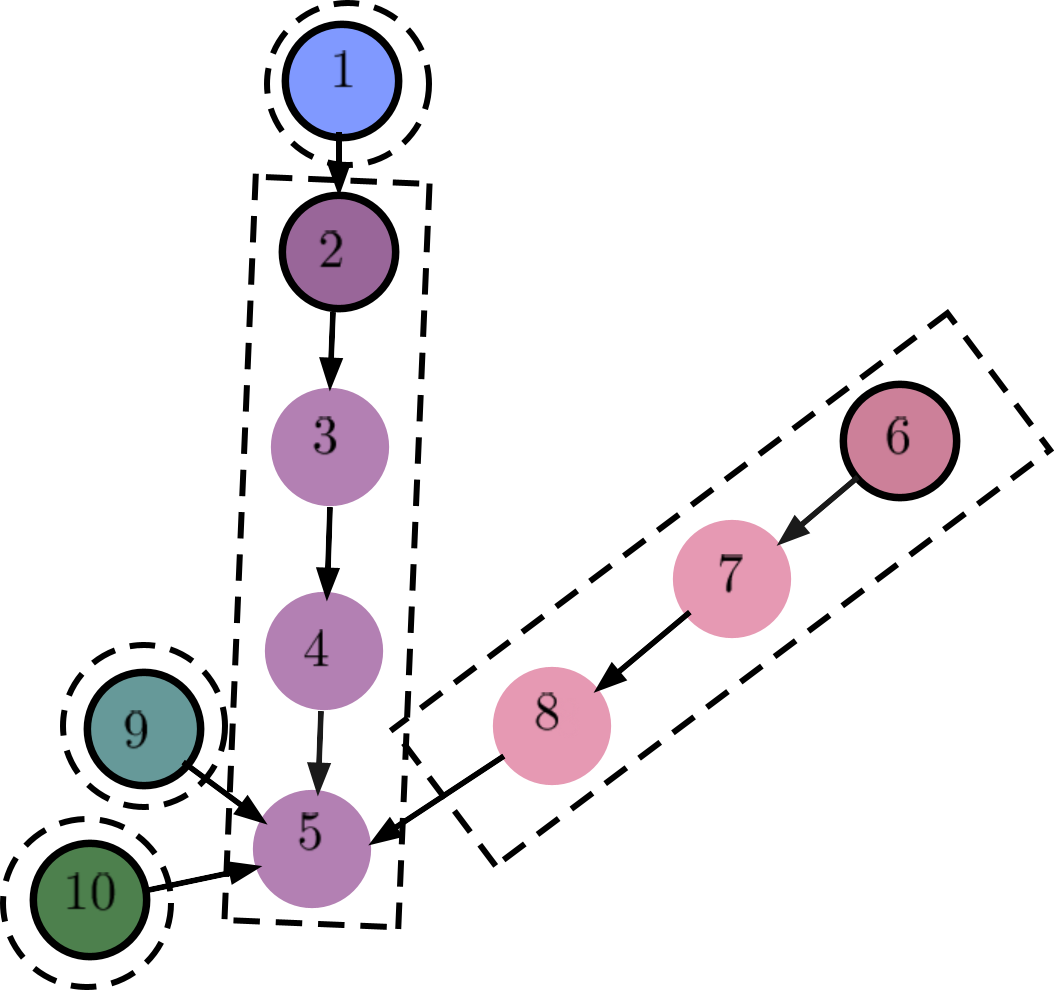}}
\caption{When roots are assumed in the sequence $[1, 6, 9, 10]$, known $2$-approximation algorithm for Tree gives lower bound $5$. But the optimal burning number of this tree is $4$.}
\label{fig:2approxdtree}
\end{figure}

Every directed tree has at least one vertex with in-degree zero and one vertex with out-degree zero.

\begin{algorithm}

\KwIn{A tree $T$ and a positive integer $b$.}
\KwOut{Resulting tree ($T'$) after $b$-cutting.}

\SetKwFunction{FMain}{}
\SetKwProg{Fn}{b-CUTTING}{:}{}
\Fn{\FMain{$T$, $b$}}{
$T' \leftarrow T $ \;
$T'' \leftarrow T$ \;

\For{$i = 1$ to $b$}
{

\ForEach{$v \in  T'$ }
{
\If{$\ID(v) = 1$ and $\OD(v)=0$}{
	$T'' \gets T'' \backslash \{v\}$\;
}

}
$T' \gets T''$\;
}	
\KwRet{$T'$}	

}
\caption{$b$-cutting}
\label{algo:bcut}
\end{algorithm}

\textbf{\BC{} process}  :
If $T$ is a directed tree,  we remove all the vertices with out-degree zero and in-degree one from the tree.  Let the resulting tree be $T^1$.  We call $T^1$ a $1$-cutting tree. Having cut vertices with out-degree zero and in-degree one of $T^1$, we get a tree $T^2$. We call $T^2$ a $2$-cutting tree. Having cut all vertices with out-degree zero and in-degree one of $T^2$,  we get a tree $T^3$.  We call $T^3$ a $3$-cutting tree. So after $b$ cutting, the resulting tree is $T^b$.  We call $T^b$ as a $b$-cutting tree of $T$.

The \BC{} process, as given in Algorithm~\ref{algo:bcut} returns a tree $T^b$. In the tree, $T^b$, vertex $v$  with out-degree zero is considered as a center of the fire.  We divide the centers into two sets, $BS$ and $BS'$.  The set $BS$ contains all the centers with in-degree $1$ and out-degree $0$ in tree $T^b$. The set $BS'$ contains all the centers with in-degree greater than one and out-degree $0$. After \BC{} if tree $T^b$ is not empty, we do \BC{} again on $T^b \backslash (BS \cup BS')$. This process is repeated until the tree becomes empty.


\BC{} returns tree $T^b$ of $T$.  $\forall c \in V(T^b)$,   if   $D^{-}(c) >1$ and $D^{+}(c) =0$ then we add  $c$ in $BS'$ and if $D^{-}(c) =1$ and $D^{+}(c) =0$, we add $c$ in $BS$.
The \GC procedure, in turn, calls \BC($T,b$) and returns the sets $BS$ and $BS'$.
\begin{lemma}
If $|BS|> b$, $b(T)> b$.
\label{lemmabs}
\end{lemma}
\begin{proof}
Let $\{c_1,c_2,c_3 \cdots c_{b+1}\}$ be the centers in $BS$. Each center has an associated directed path of length $b$. To burn all the centers and vertices in their associated paths in at most $b$ rounds, a vertex should exist that can burn at least two centers and their associated paths.

Let $c_i$ and $c_j$ be the two centers and their associated paths be $P_i$ and $P_j$ respectively. Let $u \in P_i$ and $v \in P_j$ be the farthest nodes in the $b$-length paths of $P_i$ and $P_j$ respectively. To burn $c_i$ and $c_j$ as well as $u$ and $v$ with a single source of fire, that node should be a common ancestor of $c_i$ and $c_j$. Such a common ancestor, which is at least at a distance $1$ from $c_i$ and $c_j$, will be at a distance at least $(b+1)$ from $u$ and $v$. Therefore a single source of fire can not burn two centers and their associated paths in $b$ rounds. Hence $b(T)>b$.

\end{proof}

\begin{algorithm}[!htbt]
\label{algo:lb3}
\DontPrintSemicolon
\SetKwInOut{Input}{Input}\SetKwInOut{Output}{Output}
\Input   {$T=(V,E) $ and a positive integer $b$.}
\Output  {RETURN burning sequence or -1.}
\SetKwFunction{FMain}{\GCT}
\SetKwProg{Fn}{Function}{}
\Fn{\FMain{$T$,$b$}} {
\Begin{
$i \gets b$\;
$BS \gets \phi$\;
$BS' \gets \phi$\;
\While{$i > 0$ and $|V(T)|>0$}{
$T' \gets \BC(T,b-1)$\;
\Comment{$V(T')$ is the set of all vertices in tree  $T'$.}
\ForEach{$v \in V(T')$}{
	\If{ $\OD(v)=0$}{
		\If{$\ID(v) \le 1$}{
			$append(BS, v)$\;
			\Comment{After b-cutting, Set $BS$ will contain vertices having at most one in-degree.}
		}
		\Else{
			$append(BS', v)$\;
			\Comment{After b-cutting, Set $BS'$ will contain vertices having more than one in-degree.}
		}
		
		$T \gets T \; \backslash \; \{ u: d(u,v)\le b \; \text{and} \;  u\in V \}$\;
		\Comment{Remove all vertices from $T$ which are at $b$ distance from the center of fire $v$.}
	}
}
$i \gets i-1$\;
}
\If { $|BS|>b$ or  $|BS'|>b$}{
\KwRet{$-1$}
}
\KwRet{$(BS, BS')$} \;
\Comment{$(BS, BS')$ is concatenation of the two burning sequences $BS$ and $BS'$.}
}
}
\caption{Burning number for multi-rooted directed tree $T$.}
\end{algorithm}

\begin{lemma}
\label{lemmabsd}
If $|BS'|> b$ then $b(T)> b$.
\label{lem:first}
\end{lemma}
\begin{proof}
We prove by contradiction, let $b(T)=b$.  Let $c_i$ and $c_j$ be centers in $BS'$. Hence their in-degree is greater than $1$. Each center has at least two ancestors, but there can not be more than one common ancestor between any pair of centers, as that would form an undirected cycle. As $BS' > b$, there exists at least $(b+1)$ different ancestors who need at least $(b+1)$ sources of fire to burn these ancestors. Therefore, $b(T) > b$.

\end{proof}
\GCT{} procedure is given in Algorithm~\ref{algo:lb3}. This procedure,   based on the cardinality of $|BS|$ and $|BS'|$, either  concludes that $b(T)>b$ or both $|BS|\le b$ and $|BS'|\le b$. In the latter case, the graph can be burned in $|BS| + |BS'| + b$ $\le 3b$ rounds.
\begin{theorem}
A polynomial time $3$-approximation algorithm exists for the multi-rooted directed tree.
\end{theorem}
\begin{proof}
If the algorithm returns \BADG{} for $b$, either $BS > b$ or $|BS'| > b$ and by Lemma~\ref{lemmabs} and Lemma~\ref{lemmabsd} we have $b(T) \ge b+1$. As the algorithm returns a burning sequence of length $3b+3$ for $b+1$, we get a $ 3$-approximation.

\end{proof}

Note that, for  single-rooted trees, $|BS'| = 0$,  in which case, Algorithm~\ref{algo:lb3} leads to a $ 2$-approximation for  single-rooted directed trees.

\begin{theorem}
A polynomial time $2$-approximation algorithm exists for graph burning on the single-rooted directed tree.
\label{thm:multiRootedTree}
\end{theorem}

\subsection{$1.905$-approximation for single rooted Tree}
\label{sec:dtree}

This paper presents a $ 1.905$-approximation algorithm for arborescence. The procedure is shown in Algorithm~\ref{algo:lb4}. The algorithm uses two methods \GCST{} given in Algorithm~\ref{algo:lb4}  and \merge{} given in Algorithm~\ref{algo:lb2}. The \GCST{} function returns only the set $BS$, which contains roots of the $b$ height subtrees. Note that $BS'=\phi$ because each node has at most one in-degree in an arborescence. The procedure \merge{} takes the $BS$ set and merges subtrees whose roots have the lowest common ancestor (LCA) within $\le 0.81b$ distance. The lowest common ancestor of the two roots will be the new root or the center of the fire.
\newpage
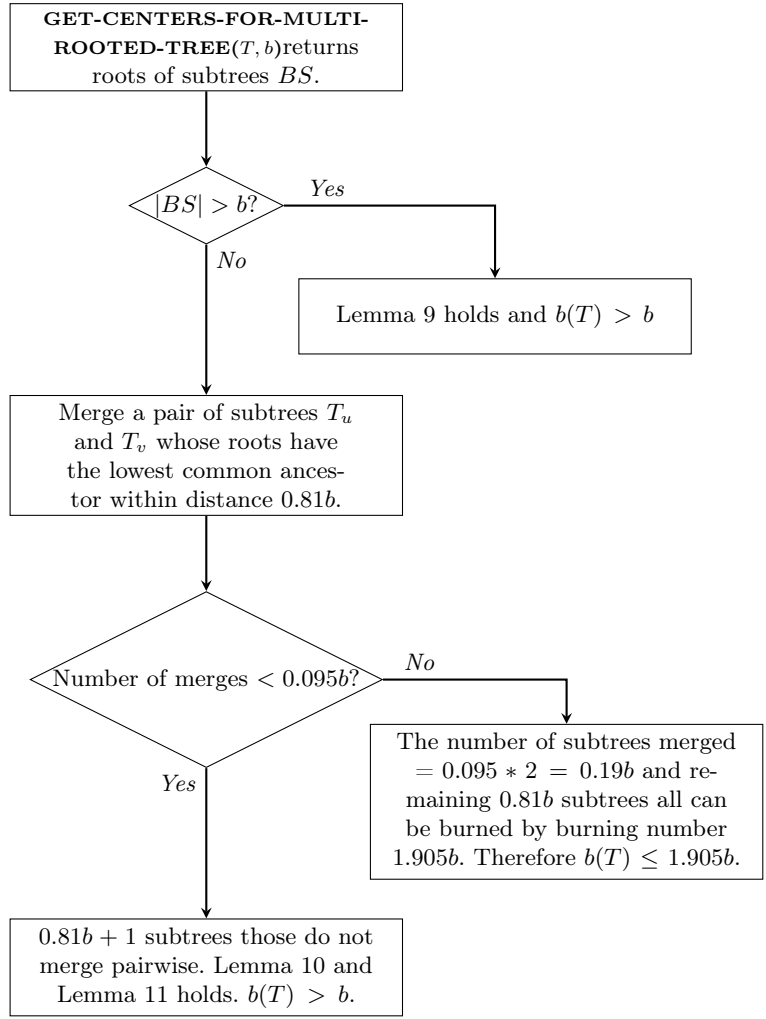
\begin{figure}[ht]
\centering
\begin{tikzpicture}[node distance=1cm]
\node[process2] (start) {\GCT returns roots of subtrees $BS$.};
\node[below = of start,decision1] (box3) {$|BS|> b$?};
\node[below right = of box3,process2] (box4) {Lemma~\ref{lem:first1} holds and $b(T)>b$};
\node[below =  2cm and -3cm of box3,process2] (box5) {Merge a pair of  subtrees $T_u$ and $T_v$ whose roots have the lowest common ancestor within  distance $0.81b$.};
\node[below = of box5,decision1] (box6) {Number of merges $<0.095b$?};
\node[below right = 0cm and 1cm of box6,process2] (box7) {The number of subtrees merged = $0.095*2 = 0.19b$ and remaining $0.81b$ subtrees all can be burned by burning number $1.905b$. Therefore $b(T)\le 1.905b$.};
\node[below =  2cm and -3cm of box6,process2] (box8) {$0.81b+1$ subtrees those do not merge pairwise. Lemma~\ref{lm:bvertex} and Lemma~\ref{lem:ele} holds. $b(T)>b$.};
\draw[arr] (start) -- (box3);
\draw[arr] (box5) -- (box6);
\draw[arr] (box6) -- (box8)node[pos=0.1,left ]{\textit{Yes}};
\draw[arr] (box6) -| (box7) node[pos=0.1,above ]{\textit{No}};
\draw[arr] (box3) -| (box4) node[pos=0.1,above ]{\textit{Yes}};
\draw[arr] (box3) -- (box5) node[pos=0.1,right]{\textit{No}};
\end{tikzpicture}
\caption{This flowchart depicts the idea of the $1.905$-approximation algorithm.}
\label{fig:merge}
\end{figure}

The main idea of the algorithm is shown as a flowchart in \figurename~\ref{fig:merge}.
\begin{lemma}
If $|BS|>b$ then $b(T)>b$.
\label{lem:first1}
\end{lemma}
\begin{proof}
The proof is similar to Lemma~\ref{lemmabs}.
\end{proof}

Suppose there are $0.81b+1$ subtrees that are not merged pairwise. We know that the  height of each subtree is $b$ except for the last subtree.  For each subtree $T_v$, a
$b$-length path $P_i$ exists from the root of $T_v$ downwards.  Take a subpath $P'_i$ of $P_i$, where $P'_i$ contains all middle $0.62b$ vertices of $P_i$; in other words,  $P'_i$ is a subpath obtained  after removing $0.19b$ vertices from both ends of path $P_i$.  Let $S$ be the union of the sets of  vertices of  the subpaths $P'_i$  obtained from the non-merged subtrees.

\begin{algorithm}[!htbt]
\label{algo:lb2}
\DontPrintSemicolon
\SetKwInOut{Input}{Input}\SetKwInOut{Output}{Output}
\Input   {\TG and a positive integer $b$.}
\Output  {Return burning sequence or -1.}
\SetKwFunction{FMain}{\merge{}}
\SetKwProg{Fn}{Function}{}
\Fn{\FMain{$T$,$b$,$BS$}} {
\Begin{
$B \gets [0, 1, 2, 3 , \cdots 1.905b]$\;
$BS_1 \gets [\;]$\;
$BS_2 \gets [\;]$\;
\ForEach{ $v \in BS$}{
$flag \gets True$\;
\ForEach{$u \in BS$}{
	\If{$v \ne u$  and $LCA\_length(T,v,u)<0.81b$}{
		\Comment{$LCA\_length(T,v,u)$ returns length of lowest common ancestor of $u$ and $v$. }
		
		\If{$\exists \; n \in B $ s.t $n \ge 1.81b$}{
			$remove(BS,u)$\;
			$remove(BS,v)$\;
			\Comment{Remove both the centers of fire $u$ and $v$   from  $BS$.}
			
			$append(BS_2, LCA\_vertex(T,v,u))$\;
			\Comment{$LCA\_vertex(T,v,u)$ returns the lowest common ancestor of $u$ and $v$.}
			$remove(B,n)$\;
			$ flag \gets False $\;
		}
	}
}
\If{$flag = True$}{
	$S = \{n \in B \; | \; n \ge b\}$\;
	\If{$|S|>0$}{
		$append(BS_1, u)$\;
		$remove(B,min(S))$\;
		\Comment{$min(S)$ returns the minimum value in the set $S$.}
	}
	\Else{
		\KwRet{$-1$}
	}
}
}
\KwRet{$(BS_1,BS_2)$}
}
}
\caption{Merge and Burn Procedure \TG.}
\end{algorithm}

\begin{algorithm}[!htbt]
\label{algo:lb4}
\DontPrintSemicolon
\SetKwInOut{Input}{Input}\SetKwInOut{Output}{Output}
\Input   {\TG and a positive integer $b$.}
\Output  {Return burning sequence $BS$, or -1 in case it fails to burn the tree $T$.}
\SetKwFunction{FMain}{GET-CENTERS-FOR-SINGLE-ROOTED-TREE}
\SetKwProg{Fn}{Function}{}
\Fn{\FMain{$T$,$b$}} {
\Begin{
$T' \gets T$\;
$BS \gets [ \; ]$\;
\While{$T' \ne \phi $}{
$T' \gets \BC(T',b-1)$\;
\Comment{$V(T')$ is the set of all vertices in tree $T'$}
\ForEach{$v \in V(T')$}{
	\If{$\ID(v) \le 1 $ and $\OD(v)=0$}{
		$append(BS, v)$\;
		\Comment{Append center of fire $v$ in the Burning sequence $BS$.}
		$T'$ $\gets$ $T'$ $\backslash$ $\{v\}$\;
	}
}
}
\KwRet{$\merge(T,b,BS)$}
}
}
\caption{Burning number for directed tree \TG.}
\end{algorithm}

\begin{figure}[H]
\centering
\def\svgwidth{\columnwidth}
\scalebox{0.10}{\includegraphics{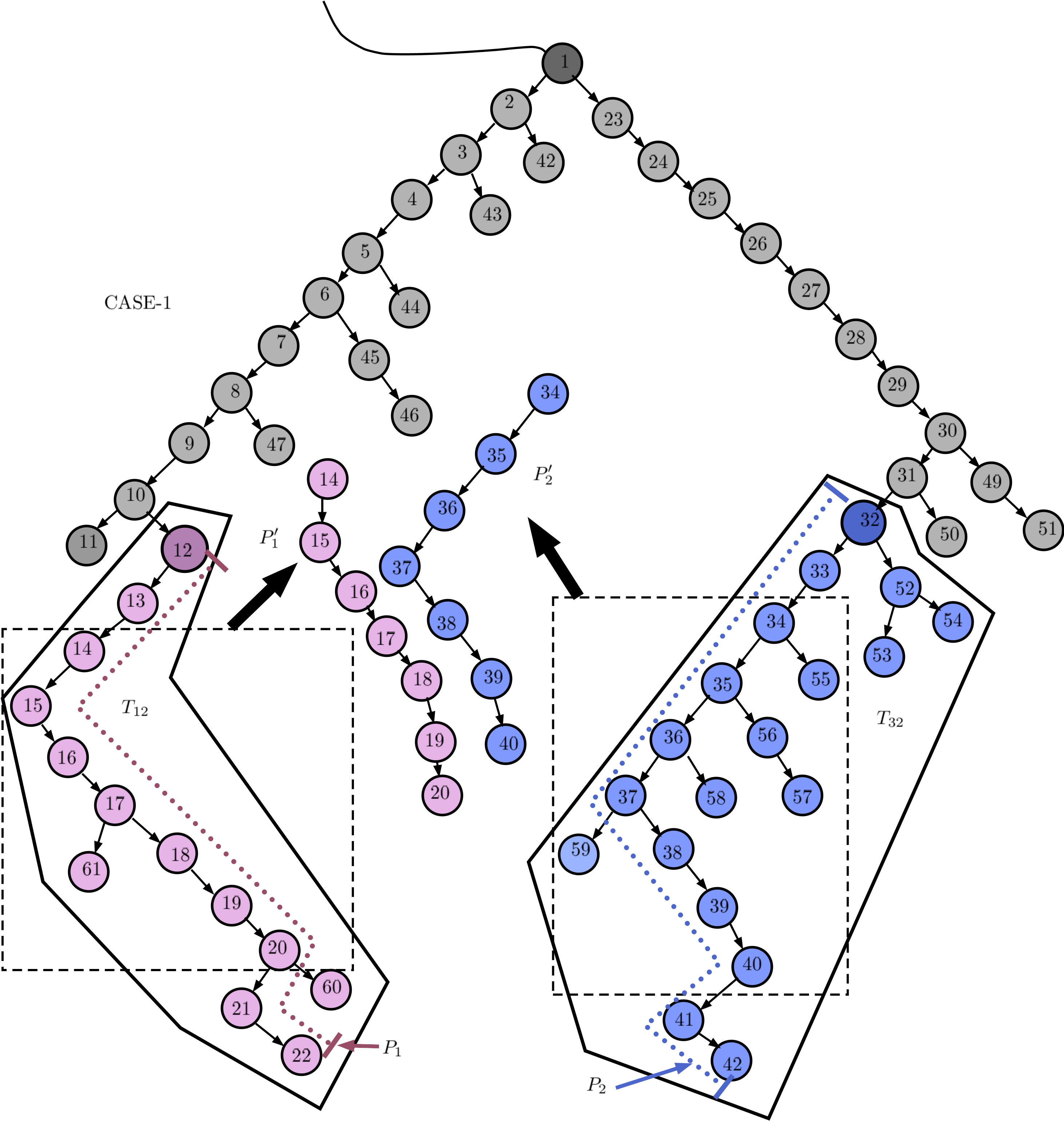}}
\caption{In this example, for $b=11$ there are two subtrees $T_{12}$ and $T_{32}$ rooted at $12$ and $32$, respectively. with paths $P_1 = [12,13,14,15,16,17,18,19,20,21,22]$ and $P_2 = [32,33,34,35,36,37,38,39,40,41,42]$. As $11*0.19 \approx 2$ ,  $P'_1 =[14,15,16,17,18,19,20]$ and $P'_2 =[34,35,36,37,38,39,40]$.}
\label{fig:case-1}
\end{figure}

\begin{figure}[H]
\centering
\def\svgwidth{\columnwidth}
\scalebox{0.10}{\includegraphics{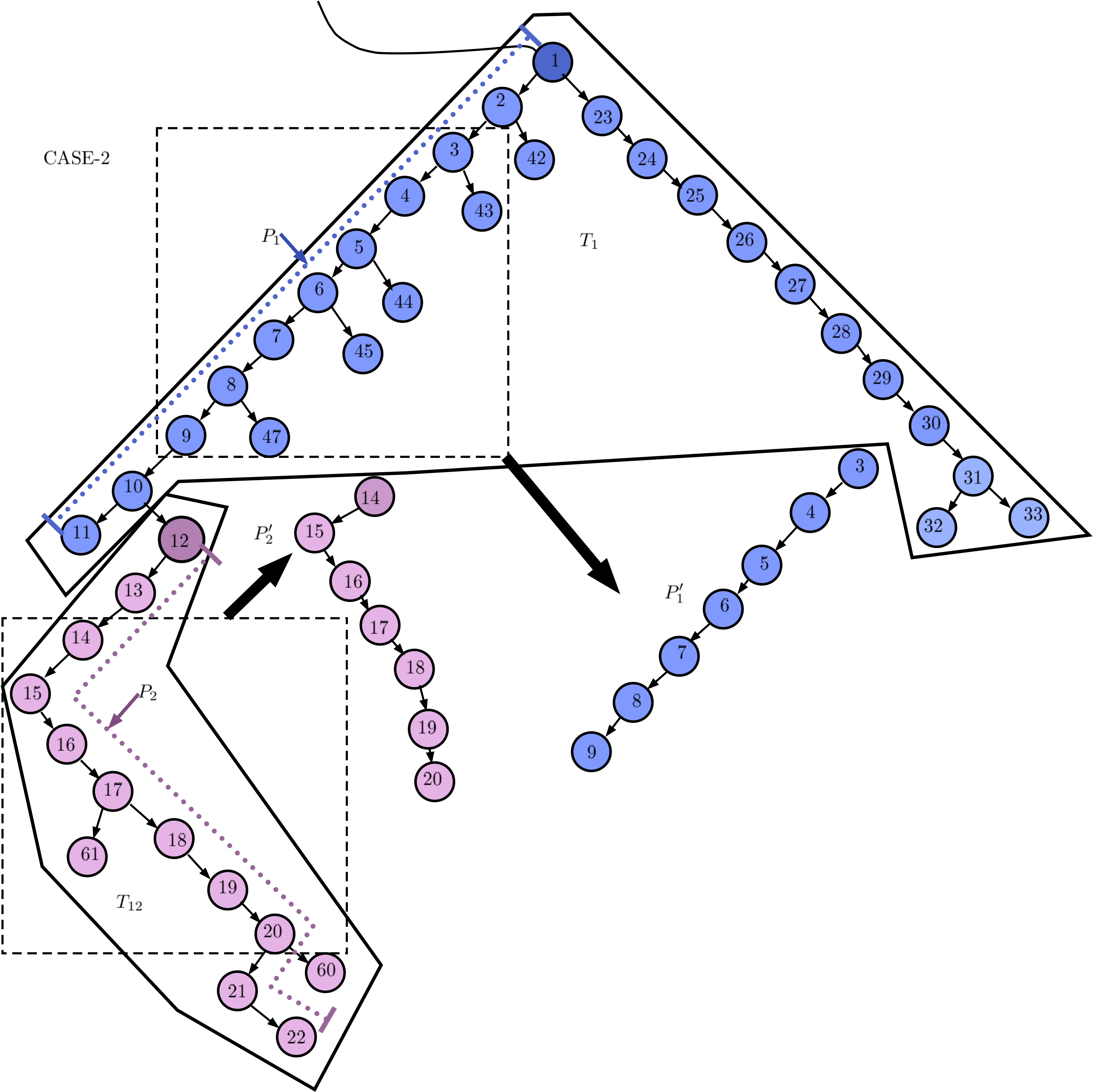}}
\caption{In this example, for $b=11$ there are two subtrees $T_1$ and $T_{12}$ at rooted $1$ and $12$, respectively. With paths $P_1 = [1,2,3,4,5,6,7,8,9,10,11]$ and $P_2 = [12,13,14,15,16,17,18,19,20,21,22]$. As  $11*0.19 \approx 2$ , $P'_1 =[3,4,5,6,7,8,9]$ and $P'_2 =[14,15,16,17,18,19,20]$.}
\label{fig:case-2}
\end{figure}

\begin{lemma}
A single source of fire with a burning range $b-1$ can burn at most $b$ vertices in $S$.
\label{lm:bvertex}
\end{lemma}
\begin{proof}

There are two cases as shown in \figurename~\ref{fig:case-1} and \figurename~\ref{fig:case-2}: first when there does not exist a directed  path from any vertex $u \in P'_i$ to $v \in P'_j$, $\forall i\ne j$. 
Secondly, when there exists a directed path either from vertex $u \in P_i$ to $v \in P_j$  or vice versa.\\

\textbf{CASE-1} If there does not exist a directed path between any two vertices $u,v$, $u \in P'_i$ and $v \in P'_j$, then  there does not exist a lowest common ancestor (LCA) of $u$ and $v$ within $b$ distance. Because if the lowest common ancestor of $u$ and $v$ is within $b$ distance, and  roots of subtrees containing $P_i$ and $P_j$ have the lowest common ancestor within $0.81b$ distance. But we know $P_i$ and $P_j$ belong to non-merged subtrees, so  roots of both the subtrees should have LCA greater than $0.81b$ distance. Therefore, a fire source with a burning range $b-1$ can burn at most $b$ vertices.\\

\textbf{CASE-2} There exists a path from $u$ to $v$ where $u \in P'_i$ and $v \in P'_j$ for $i \ne j$.
Without loss of generality, let the vertices $u$ and $v$ be closer to the roots of both subtrees.
We have $d(u,v)\ge 0.81b$ and $d(u,u') =0.62b$  where $u'$ is the deepest vertex of $P'_i$.
A fire started at any vertex in $P_i$ burns the deepest vertex $u' \in P_i$, then the top vertex $v \in P'_j$. The fire burns sequentially first the vertices of $P'_i$,  then $P'_j$, so it can burn only $b$ number of vertices.
\end{proof}


According to Lemma~\ref{lm:bvertex}, a burning range $b-1$ can burn only $b$ vertices of the set $S$. For the burning number $b$,  the sequence of the burning ranges of the sources of fire is $[b-1,b-2,b-3, \cdots,0]$. The burning number $b$ can burn at most $S_b$ vertices, where
\begin{equation} \label{eqn1}
S_b = \sum_{i= b}^{i= 1} i =  \frac{b*(b+1)}{2}
\end{equation}
The remaining task is to find the number of vertices in $S$. There are $0.81b+1$ non-merged subtrees. The number of vertices in $S$:
\begin{equation} \label{eqn2}
|S| = \sum_{i=1}^{i=0.81b+1} |P'_i| = \sum_{i=1}^{i=0.81b+1} 0.62b = (0.81b+1)*0.62b = 0.5022b^2 + 0.62b
\end{equation}

If $|S| > S_b$, then the lower bound of the burning number of the directed tree $T$ is at least $b+1$.
\begin{lemma}
If $|S|>S_b$ then $b(T)>b$.
\label{lem:ele}
\end{lemma}
\begin{proof}
A burning range $b-1$ can burn at most  $b$ vertices of the set $S$. By equation ~\ref{eqn1}, burning number $b$ can burn less than the total number of vertices in $S$ as shown in equation~\ref{eqn2}. Therefore, $b(T)>b$.
\end{proof}

\begin{theorem}
A polynomial time $1.905$-approximation algorithm exists for the single-rooted directed tree.
\label{thm:directedTree}

\end{theorem}
\begin{proof}
Algorithm uses burning ranges $\{1.905b-1, 1.905b-2 \cdots 0\}$ in the whole burning process. If the algorithm fails to burn the entire graph, Lemma~\ref{lem:ele} and Lemma~\ref{lem:first1} claim that we can not burn the whole graph in less than or equal to $b$ rounds. Therefore, on success, the algorithm returns the burning sequence of  length $1.905b$.
\end{proof}
\begin{table}[ht]
\caption{Comparison  of estimated burning numbers $b(G)$ of different approximation algorithms for randomly generated Cactus Graphs.}
\centering
\begin{tabular}{|>{\centering\arraybackslash} m{1.0cm} | >{\centering\arraybackslash}m{1.4cm} | >{\centering\arraybackslash} m{1.4cm} |>{\centering\arraybackslash}m{1.4cm} | >{\centering\arraybackslash}m{1.4cm}| >{\centering\arraybackslash} m{1.4cm}| >{\centering\arraybackslash} m{1.4cm}|}
\hline
Name & $$|V|$$ & $|E|$ & $3$-Approx~\cite{Bonato2019approx} & $2.75$-Aprrox \\ \hline
$G_{1}$ & 303 & 327 & 25 & 19 \\ \hline
$G_{2}$ & 1152 & 1223 & 37 & 35 \\ \hline
$G_{3}$ & 2186 & 2303 & 46 & 46 \\ \hline
$G_{4}$ & 3270 & 3435 & 58 & 57 \\ \hline
$G_{5}$ & 4471 & 4690 & 67 & 66 \\ \hline
$G_{6}$ & 6743 & 7012 & 73 & 71 \\ \hline
$G_{7}$ & 7824 & 8140 & 76 & 79 \\ \hline
$G_{8}$ & 9766 & 10133 & 82 & 82 \\ \hline
$G_{9}$ & 11250 & 11669 & 85 & 88 \\ \hline
$G_{10}$ & 13197 & 13658 & 94 & 90 \\ \hline
$G_{11}$ & 15327 & 15839 & 103 & 101 \\ \hline
$G_{12}$ & 17244 & 17812 & 109 & 107 \\ \hline
$G_{13}$ & 19975 & 20584 & 115 & 112 \\ \hline
$G_{14}$ & 22391 & 23059 & 130 & 115 \\ \hline
$G_{15}$ & 24011 & 24719 & 124 & 118 \\ \hline
$G_{16}$ & 26207 & 26965 & 127 & 126 \\ \hline
$G_{17}$ & 28492 & 29303 & 130 & 126 \\ \hline
$G_{18}$ & 34172 & 35083 & 142 & 137 \\ \hline
$G_{19}$ & 37759 & 38721 & 151 & 145 \\ \hline
$G_{20}$ & 39502 & 40516 & 63 & 159 \\ \hline
$G_{21}$ & 42323 & 43388 & 160 & 148 \\ \hline
$G_{22}$ & 45209 & 46315 & 166 & 159 \\ \hline
$G_{23}$ & 46857 & 48012 & 172 & 162 \\ \hline
$G_{24}$ & 48736 & 49930 & 172 & 167 \\ \hline
\end{tabular}
\label{table:cactus}
\end{table}

\begin{table}[H]

\caption{Comparison  of estimated burning numbers ($b(T)$) of different approximation algorithms for randomly Directed Tree.}
\begin{adjustbox}{angle=90}
\begin{tabular}{|>{\centering\arraybackslash} m{1.0cm} | >{\centering\arraybackslash}m{1.4cm} | >{\centering\arraybackslash} m{1.4cm} |>{\centering\arraybackslash}m{1.4cm} | >{\centering\arraybackslash}m{1.4cm}| >{\centering\arraybackslash} m{1.4cm}| >{\centering\arraybackslash} m{1.4cm}|}
\hline
Name & $|V|$ & Our $3$-Approx \\ \hline
$T_{1}$ & 3000 & 50 \\ \hline
$T_{2}$ & 5000 & 147 \\ \hline
$T_{3}$ & 7000 & 150 \\ \hline
$T_{4}$ & 9000 & 148 \\ \hline
$T_{5}$ & 13000 & 156 \\ \hline
$T_{6}$ & 15000 & 160 \\ \hline
$T_{7}$ & 17000 & 161 \\ \hline
$T_{8}$ & 19000 & 161 \\ \hline
$T_{9}$ & 23000 & 165 \\ \hline
$T_{10}$ & 25000 & 168 \\ \hline
$T_{11}$ & 27000 & 170 \\ \hline
$T_{12}$ & 29000 & 176 \\ \hline
$T_{13}$ & 33000 & 176 \\ \hline
$T_{14}$ & 35000 & 177 \\ \hline
$T_{15}$ & 37000 & 185 \\ \hline
$T_{16}$ & 39000 & 185 \\ \hline
$T_{17}$ & 41000 & 184 \\ \hline
$T_{18}$ & 43000 & 188 \\ \hline
$T_{19}$ & 45000 & 190 \\ \hline
$T_{20}$ & 47000 & 193 \\ \hline
$T_{21}$ & 49000 & 188 \\ \hline
\end{tabular}
\label{table:multiroottree}
\end{adjustbox}
\end{table}

\begin{table}[H]
\caption{Comparison  of estimated burning numbers ($b(T)$) of different approximation algorithms for randomly generated directed trees.}
\begin{threeparttable}
\begin{tabular}{|>{\centering\arraybackslash} m{3.0cm} | >{\centering\arraybackslash}m{1.4cm} | >{\centering\arraybackslash} m{1.4cm} |>{\centering\arraybackslash}m{1.4cm} | >{\centering\arraybackslash}m{1.4cm}| >{\centering\arraybackslash} m{1.4cm}| >{\centering\arraybackslash} m{1.4cm}|}
\hline
\multirow{2}{*}{Random trees} &
{|V|} &
{Our $2$-Approx} &
{Our $1.905$-Approx} \\
\hline
{$T_1$}& 1000&  32 &  27  \\
\hline
{$T_2$}& 2000&  38 &  35  \\
\hline
{$T_3$}& 3000&  44 &  42   \\
\hline
{$T_4$}& 4000& 50 & 48    \\
\hline
{$T_5$}& 5000&  50 &  48   \\
\hline
{$T_6$}& 6000&  62 & 60    \\
\hline
{$T_7$}& 7000&  58 & 56   \\
\hline
{$T_8$}& 8000& 66 &  63   \\
\hline
{$T_9$}& 9000&  68 &  65    \\
\hline
{$T_{10}$}& 10000&  72 & 69   \\
\hline
{$T_{11}$}& 11000&  74 & 71   \\
\hline
{$T_{12}$}& 12000& 74 & 73    \\
\hline
{$T_{13}$}& 13000&  76 & 73    \\
\hline
{$T_{14}$}& 14000& 78 &  75   \\
\hline
{$T_{15}$}& 15000&  80 &  77    \\
\hline
{$T_{16}$}& 16000&  82 &  79   \\
\hline
{$T_{17}$}& 17000&  84 &  81   \\
\hline
{$T_{18}$}& 18000& 90  &  86   \\
\hline		
{$T_{19}$}& 19000&  90 &  86  \\
\hline
{$T_{20}$}& 20000& 88 & 86 \\
\hline
\end{tabular}
\end{threeparttable}
\label{table:directT}
\end{table}

\subsection{Experiments and Results}
We have implemented the proposed approximation algorithms for 1) cactus graphs, 2) poly-trees, and 3) directed trees.
In the first experiment, we generate 24 cactus graphs randomly with orders ranging from 300 to 49,000. The proposed $2.75$-approximation algorithm is compared with Bonato et al.~\cite{Bonato2019approx}, and the results are shown in Table~\ref{table:cactus}. The $2.75$-approximation algorithm gives better burning numbers than the $3$-approximation algorithm.

As there are no known approximation algorithms for graph burning on directed graphs, we give the results obtained by our $3$-approximation algorithm. We generate $21$ poly-trees randomly, having vertices ranging from 3000 to 49,000. The results are shown in Table~\ref{table:multiroottree}.

Finally, in the $3rd$ experiment, the results obtained for single directed trees are shown in Table~\ref{table:directT}. We generated $20$  directed trees with the number of vertices ranging from 1000 to 20000. The implementations of our $2$-approximation algorithm (Theorem~\ref{thm:multiRootedTree}), as well as the $1.905$-approximation algorithm for directed trees (Theorem~\ref{thm:directedTree}), are compared. And the results are shown in Table~\ref{table:directT}.

It can be seen that, in all the cases, the $1.905$-approximation algorithm yields a lower burning number than the $2$-approximation algorithm.

\section{Conclusions and Future Work}
\label{sec:concl}
In this paper, we have proposed an approximation algorithm for \gb{} on  cactus graphs, poly trees, and directed trees. The approximation algorithm for  cactus graphs improves the approximation factor from $3$ to $2.75$. Further, our approximation algorithms for directed trees initiate the study of approximation algorithms for graph burning on directed graphs. The proposed approximation algorithms on directed trees have approximation factors of $3$ for the poly tree and $1.905$ for the directed tree. We have done experimentation on randomly generated 
cactus graphs and directed trees of sizes up to $49K$ to demonstrate the experimental results.  
The concepts of extracting equal height sub-trees and the merge technique may be explored for the other related problems in information diffusion. In the future, we would like to explore the approximation algorithms for the generalized \gb{} problem, namely the $k$-burning problem.

\bibliographystyle{plain}
\bibliography{myrefs}
\end{document}